\pdfoutput=1
\documentclass[10pt,pra,aps,superscriptaddress,nofootinbib,twocolumn,longbibliography]{revtex4-1}
\usepackage{amsmath,bbm}
\usepackage{amsthm}
\usepackage{amsmath}
\usepackage{latexsym}
\usepackage{amssymb}
\usepackage{float}
\usepackage{graphicx}           
\usepackage{color}
\usepackage{xcolor}
\usepackage{mathpazo}
\usepackage{comment}
\usepackage{enumitem}
\usepackage{multirow}
\usepackage{setspace}
\usepackage{subcaption}
\usepackage[colorlinks=true,linkcolor=blue,citecolor=magenta,urlcolor=blue]{hyperref}

\newcommand{\be}{\begin{equation}}
\newcommand{\ee}{\end{equation}}
\newcommand{\bea}{\begin{eqnarray}}
\newcommand{\eea}{\end{eqnarray}}
\def\squareforqed{\hbox{\rlap{$\sqcap$}$\sqcup$}}
\def\qed{\ifmmode\squareforqed\else{\unskip\nobreak\hfil
\penalty50\hskip1em\null\nobreak\hfil\squareforqed
\parfillskip=0pt\finalhyphendemerits=0\endgraf}\fi}
\def\endenv{\ifmmode\;\else{\unskip\nobreak\hfil
\penalty50\hskip1em\null\nobreak\hfil\;
\parfillskip=0pt\finalhyphendemerits=0\endgraf}\fi}

\newcommand{\tr}{\text{Tr}}
\newcommand{\I}{\mathbbm{I}}

\newcommand{\ket}[1]{|#1\rangle}
\newcommand{\bra}[1]{\langle#1|}

\newcommand{\la}{\langle}
\newcommand{\ra}{\rangle}

\makeatletter
\newtheorem*{rep@theorem}{\rep@title}
\newcommand{\newreptheorem}[2]{%
\newenvironment{rep#1}[1]{%
 \def\rep@title{#2 \ref{##1}}%
 \begin{rep@theorem}}%
 {\end{rep@theorem}}}
\makeatother

\newtheorem{thm}{Theorem}
\newreptheorem{thm}{Theorem}
\newtheorem{lemma}{Lemma}
\newtheorem{definition}{Definition}

\newtheorem{obs}{Observation}

\newtheorem{coro}{Corollary}
\newtheorem{proposition}{Proposition}

\newtheorem{example}{Example}

\begin{document}

\title{Nonlocality without entanglement in exclusion of quantum states}

\author{Satyaki Manna}
\email{mannasatyaki@gmail.com}
\affiliation{Department of Physics, School of Basic Sciences, Indian Institute of Technology Bhubaneswar, Odisha 752050, India}
\author{Anandamay Das Bhowmik}
\email{ananda.adb@gmail.com}
\affiliation{S. N. Bose National Centre for Basic Sciences, Block JD, Sector III, Salt Lake, Kolkata 700 106, India}
\affiliation{School of Physics, Indian Institute of Science Education and Research Thiruvananthapuram, Kerala 695551, India}

\begin{abstract}
    We study the task of quantum state exclusion, focusing on antidistinguishability (elimination of one state) and $2$-antidistinguishability (elimination of two states), under global measurements and local operations with classical communication (LOCC). We also introduce weak and strong notions of antidistinguishability ($2$-antidistinguishability) depending on whether all states or all $2$-tuples are exhaustively eliminated. Our results reveal striking differences between state exclusion and the more familiar task of state discrimination. In particular, we prove that a set of multipartite product states is LOCC antidistinguishable if and only if one of the parties can antidistinguish her local states alone. Consequently, LOCC antidistinguishability of product states requires no communication and is symmetric with respect to the initiating party, whereas this symmetry breaks down for one-way LOCC $2$-antidistinguishability. Most notably, we establish a manifestation of \emph{nonlocality without entanglement} in the context of state exclusion: we prove that three bipartite product states can be globally antidistinguishable while failing to be LOCC antidistinguishable, demonstrating that three is the minimal number of states required for this phenomenon. We further extend this separation to $2$-antidistinguishability and present examples exhibiting the same type of nonlocality. Finally, we provide a set of globally antidistinguishable tripartite product states that are not LOCC antidistinguishable across any bipartition, which establishes \emph{genuine nonlocality without entanglement} in this framework.
\end{abstract}

\maketitle


\section{Introduction}
 \textit{Nonlocality without entanglement} \cite{benett} is one of the most striking features of quantum theory. It arises when a set of locally preparable multipartite quantum states, although mutually orthogonal, cannot be perfectly distinguished using only local operations and classical communication (LOCC). Despite the absence of entanglement, the global structure of the composite system encodes information that remains inaccessible to purely local measurements. Owing to its fundamental importance and potential applications in quantum information processing, nonlocality without entanglement has been extensively studied \cite{benett,hardy,watrous05,Virmani_2001,vedral,ghosh_s,Halder,Bhattacharya,manna2025,gupta,ylng-5j64,2k5d-bprn}.

While this phenomenon has been predominantly investigated in the context of quantum state discrimination, it is natural to ask whether similar counterintuitive behavior can emerge in other informational tasks. In this work, we address this question by focusing on quantum state exclusion, a weaker variant of state discrimination in which the goal is to eliminate \(x\) states from a known set of \(n(>x)\) possible states \cite{Caves02,barett,Johnston2025tightbounds,bandyopadhyay,Heinosaari_2018,manna202,manna20,Manna2,yao2025,uola,webb,ji2025,ji2025conversebounds}. Specifically, we study this task for a collection of multipartite quantum states under two distinct operational paradigms: global operations, where a single party has access to the entire system, and local operations assisted by classical communication (LOCC), where spatially separated parties are restricted to local measurements on their respective subsystems.

The notions of distinguishability and exclusion of physical processes play a central role in deepening our understanding of quantum theory and, more broadly, the physical world \cite{Hellstrom,Chefles_2000,Manna2,acin2001,duan2009,watrous05,bsxv-q9x7,ghosh,njp2021,manna2026,zhu2026}. Whereas state distinguishability aims to identify which state from a known ensemble has been prepared, state exclusion seeks to rule out certain possibilities without fully identifying the prepared state. A particularly significant instance of this task is antidistinguishability\cite{Caves02,barett,Heinosaari_2018,manna20,manna202}, which involves the elimination of a single state and has been instrumental in the study of \(\psi\)-epistemic interpretations of quantum theory \cite{Pusey_2012,ray2024,bhowmik2022,leifer2014,barett,chaturvedi2021,Chaturvedi2020,ray2025}, quantum contextuality \cite{leifer}, and other foundational aspects. Beyond its conceptual relevance, quantum process exclusion has also found diverse applications in quantum information processing and communication protocols \cite{Manna1,pandit,bae2025}.

The primary goal of this paper is to study global and LOCC antidistinguishability (or, more generally, state exclusion) of sets of multipartite quantum states. By studying this, we ultimately establish \emph{nonlocality without entanglement} in the context of state exclusion. Before addressing the main problem, we provide a brief review of antidistinguishability, which concerns the elimination of a single state from a set, and its generalization to $2$-antidistinguishability, which involves the elimination of a pair of states from a given set. In this section, we introduce two notions of antidistinguishability ($2$-antidistinguishability) based on whether all states (all pairs of states) in the set are exhausted. We refer to the task of eliminating a single state (or a $2$-tuple of states) as antidistinguishability (or $2$-antidistinguishability), while the additional requirement that every state (or every $2$-tuple) in the set be eliminated at least once is referred to as strong antidistinguishability (or strong $2$-antidistinguishability).
We then describe the general protocols for global and LOCC state exclusion, followed by the presentation of our main results. Although the protocols are formulated for multipartite systems, most of the results are derived in the bipartite setting. We begin by establishing a sufficient condition for three general bipartite states to be LOCC antidistinguishable and demonstrate, by a counterexample, that this condition is necessary neither for antidistinguishability nor for strong antidistinguishability. 
Then we prove that a set of multipartite product states is LOCC antidistinguishable if and only if one party alone can antidistinguish her local states; in particular, LOCC antidistinguishability of product states requires no communication and does not depend on the commencing party.
We then investigate the role of the initiating party in one-way LOCC $2$-antidistinguishability and find that, in contrast to the single-state exclusion case, the success of the task can depend on the starter even when only product states are considered. Then we move to our main theme of \emph{nonlocality without entanglement}. We prove that there exist three bipartite product states that are antidistinguishable under global measurements but not under LOCC, thereby establishing that three is the minimal number of states required to exhibit this phenomenon. We further extend this result to the case of $2$-antidistinguishability, demonstrating again an advantage of global protocols over LOCC strategies with any number of rounds.
Finally, analogous to results known for state distinguishability, we present a manifestation of strong or genuine multipartite nonlocality in the context of antidistinguishability. In particular, we show that there exist tripartite product states that are not LOCC antidistinguishable across any bipartition, while still being globally antidistinguishable. This establishes an exclusion-based analogue of genuine multipartite nonlocality.

The paper begins with a preliminary section that provides a brief review of quantum state exclusion. We then describe both global and LOCC protocols for the elimination of quantum states, in analogy with state distinguishability. 
Our results are presented in two main sections. The first section focuses on LOCC protocols, while the second demonstrates instances of \emph{nonlocality without entanglement}. 
Finally, the conclusion summarizes the key findings and outlines open problems, along with possible directions for future research.

\section{Preliminaries}\label{prelim}
Before delving into the main objectives of the paper, it is necessary to introduce some definitions and previously established results that will be used frequently throughout our analysis. In this spirit, we begin by defining the antidistinguishability of a set of quantum states.

\subsubsection*{Antidistinguishability}

Assume we are given $n$ previously known quantum states $\{\rho_k\}_{k=1}^n$ sampled from the uniform probability distribution.
Antidistinguishability of the $n$ quantum states $\{\rho_k\}_{k=1}^n$ is defined as \cite{Heinosaari_2018}
\be  
\mathcal{A}[\{\rho_k\}_k] =  \max_{\mathcal{M}} \frac1n\sum_{k=1}^n\sum_{a\neq k}  p(a|\rho_k,\mathcal{M}),
\ee
where the optimization is over all $n$-outcome measurements $\mathcal{M}=\{M_a\}_{a=1}^n$ and $p(a|\rho_k,\mathcal{M})=\tr(\rho_kM_a)$ is the probability of obtaining outcome $a$ (interpreted as ``the state is not $\rho_a$'') when the state is $\rho_k$. As $\sum_a p(a|\rho_k,\mathcal{M})=1$, the above expression becomes
\bea\label{pA}
\mathcal{A}[\{\rho_k\}_k] 
&=& 1 - \frac{1}{n}\min_{\{M_k\}_k}\left\{\sum_{k} \ \tr(\rho_k M_k)\right\} ,
\eea 
where $\{M_k\}_k$ are the POVM (positive operator-valued measure) elements. A set is called antidistinguishable if $\mathcal{A}=1$, i.e., if there is a measurement with $\tr(\rho_kM_k)=0$ for all $k$.

\begin{example}\label{e1}
    Consider a set $\mathbbm{S}_1$ of the following three qubit states:
    \[\mathbbm{S}_1=\left\{\ket{0},\ket{v_+},\ket{v_-}\right\},
    \]
    where $\ket{v_+}=\frac12\ket{0}+\frac{\sqrt{3}}{2}\ket{1}$ and $\ket{v_-}=\frac12\ket{0}-\frac{\sqrt{3}}{2}\ket{1}$.\\
    The set $\mathbbm{S}_1$ is an antidistinguishable set. The optimum measurement has the POVM elements as follows:
    \[
    M_1=\frac23\ket{1}\bra{1}, M_2=\frac{2}{3}\ket{v_+^\perp}\bra{v_+^\perp}, M_3=\frac23\ket{v_-^\perp}\bra{v_-^\perp},
    \]
    where $\ket{v_+^\perp}=\frac{\sqrt{3}}{2}\ket{0}-\frac12\ket{1}$ and $\ket{v_-^\perp}=\frac{\sqrt{3}}{2}\ket{0}+\frac12\ket{1}$.
\end{example}

\begin{example}
We can think of another example of a set of states  $\mathbbm{S}_2$.
\[
\mathbbm{S}_2=\left\{\ket{0},\ket{1},\ket{+}\right\}.
\]
It can be easily checked that this set of states is also antidistinguishable with the measurement having POVM elements
\[
M_1=\ket{1}\bra{1}, M_2=\ket{0}\bra{0}, M_3=\mathbf{0},
\]
where $\mathbf{0}$ is the null operator.
\end{example}
In this experiment, the third outcome never occurs. But this does not contradict our definition of antidistinguishability in \eqref{pA}. Can we find an optimum measurement which has all non-null operators? 

For a qubit, an element eliminating a given pure state must be proportional to the projector onto the orthogonal state, so the only possible elements are proportional to $\ket{1}\bra{1},\ket{0}\bra{0}$ and $\ket{-}\bra{-}$. To make a valid measurement, we need
\bea
c_1 \ket{1}\bra{1} +c_2 \ket{0}\bra{0} + c_3 \ket{-}\bra{-}=\I_2,
\eea
where $c_1,c_2$ and $c_3$ are non-negative real numbers with $c_1+c_2+c_3=2$, and $\I_d$ denotes the identity operator in dimension $d$.
The off-diagonal entries force $c_3=0$, so the only solution is $c_1=c_2=1$ and $c_3=0$. Hence an optimum measurement cannot consist of three non-null operators.

From these two examples, we can make an interesting observation.
\begin{obs}\label{o1}
    If any subset of a set of states is antidistinguishable, then the set must be antidistinguishable.
\end{obs}
The antidistinguishing measurement of that subset is the optimum measurement in this case. The POVM operators $\{M_k\}_{k=1}^l$ make up the optimum measurement and the other $(n-l)$ operators are null operators. 

Needless to say, if a subset of at least two states is perfectly distinguishable, the whole set is antidistinguishable. In this case, $2\leq l\leq d$, where $d$ is the dimension of the states.

\subsubsection*{Strong antidistinguishability}
    \emph{Strong antidistinguishability} denotes the task of antidistinguishability such that all outcomes of the optimum measurement exhaust the set of states. In this case, none of the operators $\{M_k\}_k$ is a null operator.

From the definition itself, we can make the following observation.
\begin{obs}\label{o2}
    A strongly antidistinguishable set is always an antidistinguishable set, but an antidistinguishable set may not be a strongly antidistinguishable set.
\end{obs}
This last statement can be easily verified from the set $\mathbbm{S}_2$.
The set $\mathbbm{S}_2$ is not strongly antidistinguishable, though the set is antidistinguishable. But the set $\mathbbm{S}_1$ is strongly antidistinguishable. 
\begin{example}\label{e3}
    Let us take another set of states $\mathbbm{S}_3$ as follows:
    \[
\mathbbm{S}_3=\left\{\ket{0},\ket{1},\ket{+},\ket{-}\right\}.
\]
\end{example} 
From Observation \ref{o1}, this set is obviously antidistinguishable. But can this set be strongly antidistinguishable? The answer is affirmative. The POVM elements of the optimum measurement are:
\[
M_1=\frac12\ket{1}\bra{1}, M_2=\frac12\ket{0}\bra{0}, M_3=\frac12\ket{-}\bra{-},
\]
\[
M_4=\frac12\ket{+}\bra{+}.
\]

For a set of three pairwise non-orthogonal pure states, antidistinguishability and strong antidistinguishability are the same condition. Indeed, if an antidistinguishing measurement had, say, $M_3=\mathbf{0}$, then $M_1+M_2=\I$ together with $M_1\ket{\psi_1}=0$ would give $M_2\ket{\psi_1}=\ket{\psi_1}$, and hence $\la\psi_2|\psi_1\ra=\bra{\psi_2}M_2\ket{\psi_1}=0$ because $M_2\ket{\psi_2}=0$.
The necessary and sufficient conditions \cite{Caves02} for antidistinguishability (strong antidistinguishability) of three pairwise non-orthogonal pure quantum states, i.e.,
\be \label{ADSpsi123}
\mathcal{A}[\{\ket{\psi_1},\ket{\psi_2},\ket{\psi_3}\}] = 1,
\ee 
are the following:
\begin{subequations}\label{condAS}
\be
    x_1 + x_2 + x_3 < 1 
\ee
\be 
(x_1+x_2+x_3-1)^2 \geqslant 4x_1x_2x_3 ,
\ee 
\text{  where $x_1=|\la \psi_1|\psi_2\ra|^2, x_2=|\la \psi_1|\psi_3\ra|^2, x_3=|\la \psi_2|\psi_3\ra|^2$.}\\
\end{subequations}
Conditions \eqref{condAS} apply to pairwise non-orthogonal states; a triple containing an orthogonal pair is always antidistinguishable by Observation \ref{o1}.

Based on these conditions, \cite{barett} proves that a sufficient condition for the antidistinguishability of three states is that the pairwise modulus of their inner products is less than or equal to $1/2$.

Though there is no closed form to evaluate antidistinguishability of states, one can easily calculate it using semi-definite programming (SDP) which is described in \cite{Johnston2025tightbounds}. 

The same reference also states some necessary and some sufficient conditions for the antidistinguishability and non-antidistinguishability of a given set depending on the modulus of pairwise inner products of the states.  One useful sufficient condition is as follows:
For any set of states $\{\ket{\psi_i}\}_{i=1}^n$, if
\be\label{scond}
|\la\psi_i|\psi_j\ra|\leqslant\frac{1}{\sqrt{2}}\sqrt{\frac{n-2}{n-1}},
\ee
for all $1\leqslant i\neq j\leqslant n$, then the set is antidistinguishable.

\subsubsection*{$2$-Antidistinguishability}
Similarly, we can define the elimination of more than one quantum state. $2$-antidistinguishability denotes the task of eliminating at least two states from a set at each outcome of the measurement.

Assume we are given $n$ a priori known quantum states $\{\sigma_k\}_{k=1}^n$ sampled from the uniform probability distribution. Let the outcomes $a$ of the measurement $\mathcal{F}=\{F_a\}_a$ be labelled by the unordered pairs $a=\{k,l\}$, $k<l$; outcome $a$ is interpreted as ``the state is neither $\sigma_k$ nor $\sigma_l$''. The cardinality of the set of all pairs is $^nC_2=\frac{n(n-1)}{2}$. The $2$-antidistinguishability of the ensemble is defined as
\be\label{pA2}
 \mathcal{A}_2[\{\sigma_k\}_k] =\max_{\mathcal{F}} \frac1n \sum_{k=1}^n\ \sum_{a\,:\,k\notin a} p(a|\sigma_k,\mathcal{F}),
\ee
where $p(a|\sigma_k,\mathcal{F})=\tr(\sigma_kF_a)$. Since $\sum_a p(a|\sigma_k,\mathcal{F})=1$ and each pair $a=\{k,l\}$ contains exactly the indices $k$ and $l$, this simplifies to
\bea\label{pA2b}
&&\mathcal{A}_2[\{\sigma_k\}_k]\nonumber\\
&=& 1 - \min_{\{F_a\}_a} \frac1n\left\{ \sum_{k=1}^n\sum_{l> k}\tr((\sigma_k+\sigma_{l})F_{k,l})\right\}.
\eea
For better clarity, let us consider the following example. 

\begin{example}\label{e4}
    Let us consider a set of states in $\mathbbm{C}^2\otimes\mathbbm{C}^2$ of the type studied in \cite{webb}:
    \[
    \mathbbm{S}_4=\{\ket{x}\ket{x},\ket{x}\ket{y},\ket{y}\ket{x},\ket{y}\ket{y}\},
    \]
    where $\ket{x}=\cos\theta_0\ket{0}+\sin\theta_0\ket{1}$ and $\ket{y}=\cos\theta_0\ket{0}-\sin\theta_0\ket{1}$ with $\tan\theta_0=\sqrt{\sqrt{2}-1}$, i.e., $\cos2\theta_0=\sqrt2-1$.
    
    This set is $2$-antidistinguishable. 
    The measurement operators are $\ket{\phi_{ij}}\bra{\phi_{ij}}$, formed with the following non-orthogonal, unnormalized vectors:
\[   
\ket{\phi_{12}} = 
\frac{1}{\sqrt{2}}\big(\sqrt{\sqrt{2}-1}\ket{00} - \ket{10}\big),
\]
\[
\ket{\phi_{13}}=\frac{1}{\sqrt{2}}\big(\sqrt{\sqrt{2}-1}\ket{00} - \ket{01}\big),
\]
\[
\ket{\phi_{24}}=
\frac{1}{\sqrt{2}}\big(\sqrt{\sqrt{2}-1}\ket{00} + \ket{01}\big),
\]
\[
\ket{\phi_{34}}=
\frac{1}{\sqrt{2}}\big(\sqrt{\sqrt{2}-1}\ket{00} + \ket{10}\big),
\]
\[
\ket{\phi_{23}}=
\frac{1}{\sqrt{2}}\big((\sqrt{2}-1)\ket{00} + \ket{11}\big) ,
\]
\[
\ket{\phi_{14}} =
\frac{1}{\sqrt{2}}\big((\sqrt{2}-1)\ket{00} - \ket{11}\big).
\]
One verifies that $\sum_{i<j}\ket{\phi_{ij}}\bra{\phi_{ij}}=\I_4$ and, using $\tan^2\theta_0=\sqrt2-1$, that $\ket{\phi_{ij}}$ is orthogonal to the $i$'th and $j$'th states of $\mathbbm{S}_4$. Hence the outcome $\phi_{ij}$ eliminates the $i$'th and $j$'th states. This construction follows \cite{webb}. Note that this explicit measurement works only for $\theta=\theta_0$; for other values of $\theta$ in the range of Theorem \ref{th6} below (for instance $\theta=5\pi/24$) the set is still strongly $2$-antidistinguishable, but with a different measurement.
\end{example}
\begin{example}
Let us take another example of set $\mathbbm{S}_5$ in $\mathbbm{C}^3$:
\[
\mathbbm{S}_5=\big\{\ket{0},\ket{1},\ket{2},\frac{1}{\sqrt{3}}(\ket{0}+\ket{1}+\ket{2})\big\}.
\]
  This set is $2$-antidistinguishable by the measurement with POVM elements
  \[
  F_{12}=\ket{2}\bra{2}, F_{13}=\ket{1}\bra{1}, F_{23}=\ket{0}\bra{0}, 
  \]
  \[
  F_{14}=F_{24}=F_{34}=\mathbf{0}.
  \]
\end{example}
In the same analogy of $\mathbbm{S}_2$, we ask the same question here. Can we find a six-outcome measurement with no null element and execute the same task simultaneously? The answer is negative.

We can consider all the necessary elements which eliminate all the pairs from the set. As the states are in $\mathbbm{C}^3$, for each pair there is a unique (up to scaling) element which is orthogonal to both states. Considering all such elements, we can write the following equation:
\bea
&&g_1\ket{0}\bra{0}+g_2\ket{1}\bra{1}+g_3\ket{2}\bra{2}+\frac{g_4}{2}\ket{1-2}\bra{1-2}\nonumber\\
&&+\frac{g_5}{2}\ket{0-2}\bra{0-2}+\frac{g_6}{2}\ket{0-1}\bra{0-1}=\I_3,
\eea
where $\ket{i-j}\equiv\ket{i}-\ket{j}$, all $g_i\geq 0$ and $\sum_i g_i=3$. The off-diagonal entries force $g_4=g_5=g_6=0$, so the only valid solution of the above equation is $g_4=g_5=g_6=0$ and $g_1=g_2=g_3=1$.

Before going to the next subsection, let us point out an obvious observation.
\begin{obs}\label{o3}
    A $2$-antidistinguishable set is an antidistinguishable set.
\end{obs}
This observation is trivial from the definition.

\subsubsection*{Strong $2$-antidistinguishability}
    The task is similar to strong antidistinguishability. In this case, the optimum measurement must exhaust all possible $2$-tuples. That means none of the POVM operators $\{F_{k,l}\}_{k,l>k}$ is a null operator.
\begin{obs}\label{o4}
     From the definition, we can infer that a strongly $2$-antidistinguishable set is always a $2$-antidistinguishable set, but a $2$-antidistinguishable set may not be a strongly $2$-antidistinguishable set.
\end{obs}
The set $\mathbbm{S}_4$ is strongly $2$-antidistinguishable but the set $\mathbbm{S}_5$ is not strongly $2$-antidistinguishable, though it is $2$-antidistinguishable.

One can consider similar $x$-antidiscrimination, which is the elimination of $x$ states from a set in one go. If $x=n-1$, the task is nothing but the distinguishability of the set of states. In this work, we only consider antidistinguishability and $2$-antidistinguishability.

\section{Global and LOCC antidistinguishability}
Now, we present the general protocols for global and LOCC antidistinguishability of a set of quantum states. Let $\mathbbm{S}=\{\ket{\psi_i}\}_{i=1}^k=\{\otimes_{j=1}^n\ket{\psi^{(j)}_i}\}_{i=1}^k$ be a set of multipartite pure states on $\mathcal{H}=\bigotimes_{j=1}^n\mathcal{H}_j$ for $n\geqslant 2$. In principle, the dimension of each $\mathcal{H}_j$ can be different.
\begin{definition}
    \textbf{Global antidistinguishability ($2$-antidistinguishability)}
    needs a protocol which gives access to the whole set $\mathbbm{S}$ to a single party. The party chooses a suitable measurement to antidistinguish ($2$-antidistinguish) that set of states.
\end{definition}
    It is just the antidistinguishability ($2$-antidistinguishability) of a given set of states described in the previous section.
     Mathematically, we can define it as the same linear function defined in \eqref{pA} and \eqref{pA2} depending on the number (one or two) of the eliminated states. In this case, the extra condition is that the states are pure and multipartite.
    
    Similarly, we can define \textbf{strong global antidistinguishability ($2$-antidistinguishability)} as in the previous section with only pure and multipartite states.  
\begin{definition}
        \textbf{LOCC antidistinguishability ($2$-antidistinguishability)} of a set of states starts with the $j$'th party performing a measurement on her subsystem and communicating the outcome to another party, where $j$ is arbitrary. Based on the received information, that party performs a measurement on his subsystem and communicates the result to the next party. This procedure continues until all parties have participated. At each run of the protocol, at least one (two) state(s) from the set must be eliminated. The last party may communicate back to any other party, and the same procedure can be continued. However, the parties have access to only a single copy of the unknown state. 
\end{definition}

LOCC antidistinguishability (LOCC $2$-antidistinguishability) starts with the $j$'th party executing a POVM $\{M_j^\alpha\}_{\alpha}$, $\alpha=1,2,\cdots,r$ on the $j$'th subsystem. Each POVM satisfies $\sum_{\alpha=1}^r M_j^\alpha=\I$, and each element admits the Kraus form $M_j^\alpha=(F_j^\alpha)^\dagger F_j^\alpha$. The probability of getting outcome $\alpha$ for the input state $\ket{\psi_i}$ is $p_\alpha =\bra{\psi_i}\I\otimes\cdots\otimes M^\alpha_j\otimes\cdots\otimes\I\ket{\psi_i}$ and the post-measurement state is $\frac{1}{\sqrt{p_\alpha}}(\I\otimes\cdots\otimes F^\alpha_j\otimes\cdots\otimes\I)\ket{\psi_i}$. A state with $p_\alpha=0$ is eliminated by the outcome $\alpha$. The measurement $\{M_j^\alpha\}$ needs to be antidistinguishability ($2$-antidistinguishability) preserving: for every outcome $\alpha$, the states eliminated by $\alpha$, together with those eliminated by the remaining part of the protocol acting on the post-measurement states $\{\frac{1}{\sqrt{p_\alpha}}(\I\otimes\cdots\otimes F^\alpha_j\otimes\cdots\otimes\I)\ket{\psi_i}\}_{i}$, must number at least one (two) along every subsequent sequence of outcomes.

If the starter of the protocol can choose a measurement such that at every outcome, she eliminates at least one (two) state(s) from the ensemble, then the states are antidistinguishable ($2$-antidistinguishable) and this case does not need any communication.

    

\begin{definition}
    \textbf{Strong LOCC antidistinguishability ($2$-antidistinguishability)} of a set of states denotes the task of LOCC antidistinguishability ($2$-antidistinguishability) such that the optimum protocol exhausts the elimination of all the states (all the $2$-tuples). 
\end{definition}

First, we describe this protocol for antidistinguishability. Parties start with the set $\mathbbm{S}$. At each run of the protocol, they antidistinguish the set $\mathbbm{S}^1$, which obviously satisfies $\mathbbm{S}^1\subseteq\mathbbm{S}$. Suppose the protocol consists of $w$ number of runs. If each run necessitates the antidistinguishability of the subset $\{\mathbbm{S}^w\}_w$, strong LOCC antidistinguishability requires that $\bigcup_w\mathbbm{S}^w=\mathbbm{S}$. Note that the lowest cardinality of the set $\mathbbm{S}^w$ is  $1$ for all $w$. In other words, completing all the runs of the protocol, the parties need to antidistinguish the set $\mathbbm{S}$ such that each state is eliminated at least once in some run.

In the previous section, we wrote that in antidistinguishability, the optimum measurement may contain a null element, whereas in strong antidistinguishability, this is not the case. For multipartite states, in a local protocol, one can think as follows. Suppose, exhausting all $w$ runs, we see that $\bigcup_w\mathbbm{S}^w\subset\mathbbm{S}$. That means the set is not strongly LOCC antidistinguishable. We can think of the existence of null operators in some party's measurement and associate these null operators to the states which are not part of $\bigcup_w\mathbbm{S}^w$ but part of $\mathbbm{S}$. Those outcomes will not occur. 

Similarly, we are going to explain strong $2$-antidistinguishability. We can define a set $\mathbbm{S}'$ which consists of all the pairs from the set $\mathbbm{S}$. At each run of the protocol, parties eliminate the subset $\mathbbm{S}^v$, where $v$ labels the runs. Here, the lowest cardinality of all the sets $\mathbbm{S}^v$ is $2$, for all $v$. There may be some $\mathbbm{S}^{v=v*}$ whose cardinality is greater than $2$. In this case, we define the set $\overline{\mathbbm{S}}^{v=v*}$ which consists of all the pairs of $\mathbbm{S}^{v=v*}$. If the cardinality is $2$ for some $v=v'$, we can trivially say $\mathbbm{S}^{v=v'}=\overline{\mathbbm{S}}^{v=v'}$. Strong $2$-antidistinguishability requires that $\bigcup_v\overline{\mathbbm{S}}^v=\mathbbm{S}'$. In other words, completing all the runs of the protocol, the parties need to $2$-antidistinguish the set $\mathbbm{S}$ such that each pair is eliminated at least once in some run.

One can check that all the Observations \ref{o1}, \ref{o2}, \ref{o3} and \ref{o4} are true for these protocols too.

It is useful to introduce a simplified subclass of local protocols, namely \textit{one-way LOCC}. Here, the parties act in a fixed order, each performing a single local measurement conditioned on the classical outcomes communicated by the preceding parties. A protocol without any communication, in which every party measures independently and the outcomes are combined at the end, is a special case.

 In a general multi-round protocol, any party can open with a trivial measurement (a single outcome with POVM element $\I$) and thereby pass the first nontrivial move to another party. Hence the question of whether a task depends on which party starts is meaningful only within one-way LOCC, and all statements about the commencing party below refer to one-way protocols.

\subsubsection*{Tree representation} Any finite-round LOCC protocol can be represented by a finite rooted tree. Each internal node $v$ is assigned to a party $j(v)$, who applies a quantum instrument $\{\Phi_{v,x}\}_x$ to her subsystem, i.e., a collection of completely positive maps (possibly involving local ancillas) whose sum is trace preserving. The children of $v$ are labelled by the outcomes $x$, and the party and instrument at a node may depend on all earlier outcomes. For a leaf $\ell$, let $E^{(j)}_\ell$ be the effect on $\mathcal{H}_j$ obtained by composing, in the Heisenberg picture, the maps of party $j$ along the path from the root to $\ell$ (with $E^{(j)}_\ell=\I_j$ if party $j$ does not act on this path). Since all operations are local, the leaf $\ell$ corresponds to the product effect $\bigotimes_jE^{(j)}_\ell$, and for a product state $\ket{\psi_i}=\otimes_j\ket{\psi_i^{(j)}}$,
\be\label{leafprob}
p(\ell|\psi_i)=\prod_{j=1}^n\bra{\psi_i^{(j)}}E^{(j)}_\ell\ket{\psi_i^{(j)}}.
\ee
Hence the leaf $\ell$ eliminates $\ket{\psi_i}$ if and only if $E^{(j)}_\ell\ket{\psi_i^{(j)}}=0$ for some $j$. We denote by $K_j(\ell)=\{i: E^{(j)}_\ell\ket{\psi_i^{(j)}}=0\}$ the set of states eliminated at $\ell$ through party $j$, so that $\ell$ eliminates exactly the states in $\bigcup_jK_j(\ell)$. A protocol achieves LOCC antidistinguishability ($2$-antidistinguishability) if and only if $|\bigcup_jK_j(\ell)|\geq1$ ($\geq2$) for every leaf $\ell$.

Finally, we say that a list of local vectors $\{\ket{\chi_i}\}_{i=1}^k$ (repetitions allowed) is antidistinguishable if there is a POVM each of whose elements annihilates at least one $\ket{\chi_i}$. After grouping the elements according to an annihilated vector, this agrees with the definition given in Sec. II.

\begin{thm}\label{l5}
Let $\mathbbm{S}=\{\otimes_{j=1}^n\ket{\psi_i^{(j)}}\}_{i=1}^k$ be a set of product states such that, for every party $j$, the local list $\{\ket{\psi^{(j)}_i}\}_{i=1}^k$ is not antidistinguishable. Then every finite-round LOCC protocol has a leaf $\ell$ with $K_j(\ell)=\emptyset$ for all $j$, i.e., an outcome that occurs with nonzero probability for every state in $\mathbbm{S}$.
\end{thm}
\begin{proof}
We first note an identity. Fix a party $j$ and a \emph{selection} $\sigma$, i.e., a choice of one outcome at every node not assigned to $j$. Let $L_\sigma$ be the set of leaves that can be reached when the outcomes at the nodes not assigned to $j$ follow $\sigma$, while all outcomes are allowed at the nodes of party $j$. Then
\be\label{sel}
\sum_{\ell\in L_\sigma}E^{(j)}_\ell=\I_j .
\ee
This follows by induction from the leaves to the root: at a node of party $j$ the sum over its children reproduces $\sum_x\Phi^\dagger_{v,x}(\I)=\I$ by trace preservation, while at any other node only the selected child contributes and the effect of party $j$ is unchanged. Since the local list of party $j$ is not antidistinguishable, the POVM $\{E^{(j)}_\ell\}_{\ell\in L_\sigma}$ must contain an element that annihilates none of the vectors $\ket{\psi_i^{(j)}}$. Thus, for every selection $\sigma$ there is a leaf $\ell\in L_\sigma$ with $K_j(\ell)=\emptyset$.

We now prove the lemma by induction on the number $m$ of parties that act in the protocol. For $m=1$ the claim is the statement just proved, with the empty selection. For $m\geq2$, single out one acting party, say Alice ($j=1$), and consider the following finite game of perfect information on the tree, used purely as a combinatorial device: Alice chooses the outcomes at her nodes, a coalition of the remaining parties chooses the outcomes at all other nodes, and the coalition wins if the play ends at a leaf $\ell$ with $K_1(\ell)\neq\emptyset$. A strategy of the coalition is a selection $\sigma$ in the above sense, and we have shown that against any $\sigma$ Alice can reach a leaf with $K_1(\ell)=\emptyset$. Hence the coalition has no winning strategy, and by Zermelo's theorem (finite games of perfect information without draws are determined) Alice has a strategy $\tau$, i.e., a choice of one outcome at each of her nodes, such that $K_1(\ell)=\emptyset$ for every leaf $\ell\in L_\tau$ reachable when Alice follows $\tau$. Contracting every node of Alice to the child selected by $\tau$ yields a tree for the remaining $m-1$ parties, with leaf set $L_\tau$ and unchanged effects $E^{(j)}_\ell$ for $j\geq2$. By the induction hypothesis, this tree has a leaf $\ell^*$ with $K_j(\ell^*)=\emptyset$ for all $j\geq2$, and since $\ell^*\in L_\tau$, also $K_1(\ell^*)=\emptyset$.
\end{proof}

Theorem \ref{l5} concerns leaves at which no state at all is eliminated, and it therefore settles LOCC antidistinguishability of product states completely (Theorem \ref{th2} below). For $2$-antidistinguishability, a leaf may eliminate one state through Alice's effect and another through Bob's effect, so communication can become indispensable; an explicit example is given in Theorem \ref{th3}. Theorem \ref{l5} still applies to $2$-antidistinguishability whenever every nonempty local elimination removes at least two states, as in Theorem \ref{th6}.

\section{Results on LOCC protocols}\label{results}

In this section, we exhibit different results regarding LOCC antidistinguishability and LOCC $2$-antidistinguishability.
First, we comment on conditions under which three bipartite states are LOCC antidistinguishable.

Alice and Bob know the precise form of the states they share. Such three states can be written as follows:
\bea\label{gen_st}
\ket{\psi}&=&a_0\ket{0}_A\ket{\alpha_0}_B+\cdots+a_{d-1}\ket{d-1}_A\ket{\alpha_{d-1}}_B\nonumber\\
\ket{\phi}&=&b_0\ket{0}_A\ket{\beta_0}_B+\cdots+b_{d-1}\ket{d-1}_A\ket{\beta_{d-1}}_B\nonumber\\
\ket{\zeta}&=&c_0\ket{0}_A\ket{\gamma_0}_B+\cdots+c_{d-1}\ket{d-1}_A\ket{\gamma_{d-1}}_B,
\eea
where $\{\ket{0},\cdots,\ket{d-1}\}$ forms an orthonormal basis, the set of states is antidistinguishable but not pairwise orthogonal, and $a_i,b_i,c_i\geq 0$ (phases are absorbed into the normalized vectors $\ket{\alpha_i},\ket{\beta_i},\ket{\gamma_i}$, which are defined whenever the corresponding coefficient is nonzero). By normalization, $\sum_i |a_i|^2=\sum_i |b_i|^2=\sum_i |c_i|^2 =1$. As the states are antidistinguishable, we can write according to \eqref{condAS}, 
\bea\label{newcon}
&& y_1+y_2+y_3 < 1,\\
&& (y_1+y_2+y_3-1)^2\geqslant 4y_1y_2y_3,
\eea
where $y_1=|\la\psi|\phi\ra|^2=|\sum_{i=0}^{d-1}a_i^*b_i\la \alpha_i|\beta_i\ra|^2, y_2=|\la\psi|\zeta\ra|^2=|\sum_{i=0}^{d-1}a_i^*c_i\la \alpha_i|\gamma_i\ra|^2$ and $y_3=|\la\zeta|\phi\ra|^2=|\sum_{i=0}^{d-1}c_i^*b_i\la \gamma_i|\beta_i\ra|^2$.
\begin{thm}\label{th1}
    If Alice starts the protocol, a sufficient condition for the states in \eqref{gen_st} to be LOCC antidistinguishable is that, for every $i$ with $a_ib_ic_i\neq0$, the triple $\{\ket{\alpha_i},\ket{\beta_i},\ket{\gamma_i}\}$ is antidistinguishable. This condition, however, is not necessary. 
\end{thm}
\begin{proof}
    Alice measures her subsystem in the basis $\{\ket{0},\cdots,\ket{d-1}\}$ and communicates the outcome $i$ to Bob. If $a_ib_ic_i=0$, at least one of the three states has zero probability of producing the outcome $i$ and is thereby eliminated. Otherwise, the post-measurement states of Bob are $\ket{\alpha_i}$, $\ket{\beta_i}$ and $\ket{\gamma_i}$, and Bob performs the measurement that antidistinguishes this triple. This proves sufficiency.

    To see that the condition is not necessary, consider the product states
\be\label{trineB}
\ket{0}\ket{0},\quad \ket{v_+}\ket{0},\quad \ket{v_-}\ket{0},
\ee
with $\ket{v_\pm}$ as in Example \ref{e1}. Note that, these states can be written in the form of \eqref{gen_st}. Alice alone antidistinguishes them with the measurement of Example \ref{e1}, so they are LOCC antidistinguishable without any communication. However, in the computational basis the amplitudes of all three states on $\ket{0}_A$ are nonzero, while Bob's corresponding triple $\{\ket{0},\ket{0},\ket{0}\}$ is not antidistinguishable. In fact, the condition fails for every orthonormal basis $\{\ket{e_0},\ket{e_1}\}$ of Alice: if $\ket{e_0}$ is orthogonal to none of the states $\ket{0},\ket{v_\pm}$, the outcome $e_0$ has three nonzero amplitudes; if $\ket{e_0}$ is orthogonal to one of them, then $\ket{e_1}$ is parallel to that state and hence, as the three states are pairwise non-orthogonal, orthogonal to none of them. In either case some outcome has three nonzero amplitudes together with Bob's triple $\{\ket{0},\ket{0},\ket{0}\}$.
\end{proof}
The statement of the last theorem about non-necessity is applicable to strong antidistinguishability too. 

\begin{lemma}\label{l2}
    If Alice starts the protocol, the states in \eqref{gen_st} can be strongly LOCC antidistinguishable even if the condition of Theorem \ref{th1} fails in every basis of Alice.
\end{lemma}
\begin{proof}
    The states \eqref{trineB} are strongly LOCC antidistinguishable. Alice's measurement in Example \ref{e1} has three non-null elements, which eliminate $\ket{0}\ket{0}$, $\ket{v_+}\ket{0}$ and $\ket{v_-}\ket{0}$, respectively. As shown in the proof of Theorem \ref{th1}, the condition fails for these states in every basis of Alice.
\end{proof}

Theorem \ref{th1} also applies to entangled states. As an example, consider
\bea\label{necessary}
\ket{\psi_1} &=& \ket{0}\ket{+}\nonumber\\
\ket{\psi_2} &=& \frac{1}{\sqrt{2}}(\ket{0}\ket{0}+\ket{1}\ket{0})\nonumber\\
\ket{\psi_3} &=& \frac{1}{\sqrt{2}}(\ket{01}+\ket{10}).
\eea
Comparing with \eqref{gen_st}, we have $\ket{\alpha_0}=\ket{+},\ket{\beta_0}=\ket{\beta_1}=\ket{\gamma_1}=\ket{0}$, $\ket{\gamma_0}=\ket{1}$ and $a_1=0$. Suppose Alice measures in the $\{\ket{0},\ket{1}\}$ basis. If the second outcome clicks, Alice eliminates $\ket{\psi_1}$. If the first outcome occurs, Bob's triple $\{\ket{+},\ket{0},\ket{1}\}$ is antidistinguishable; Bob measures in the computational basis and eliminates $\ket{\psi_3}$ (outcome $0$) or $\ket{\psi_2}$ (outcome $1$). Since every state is eliminated in some run, the states \eqref{necessary} are even strongly LOCC antidistinguishable.
They are also LOCC antidistinguishable irrespective of the commencing party. 
 If Bob starts the protocol, Bob executes the measurement in the computational basis $\{\ket{0},\ket{1}\}$. If the first outcome occurs, upon communication, Alice measures in the same basis and eliminates either $\ket{\psi_3}$ (outcome $0$) or $\ket{\psi_1}$ (outcome $1$). If Bob gets the second outcome, he eliminates $\ket{\psi_2}$ straightforwardly.  
 
This raises the question of whether there exists a set of states whose (one-way) LOCC antidistinguishability depends on the choice of the initiating party. Our answer is negative for any set of multipartite product states, as a consequence of the following characterization.

\begin{thm}\label{th2}
    A set of multipartite product states $\{\otimes_{j}\ket{\psi^{(j)}_i}\}_{i=1}^k$ is LOCC antidistinguishable (with any finite number of rounds) if and only if, for at least one party $j$, the local list $\{\ket{\psi_i^{(j)}}\}_{i=1}^k$ is antidistinguishable. In particular, if a set of product states is LOCC antidistinguishable, it is antidistinguishable without any communication.
\end{thm}
\begin{proof}
If the local list of party $j$ is antidistinguishable, party $j$ performs the corresponding measurement; each of its elements annihilates some $\ket{\psi_i^{(j)}}$ and therefore eliminates $\ket{\psi_i}$. Conversely, if no local list is antidistinguishable, Lemma \ref{l5} shows that every LOCC protocol has an outcome eliminating no state, so the set is not LOCC antidistinguishable.
\end{proof}

\begin{coro}\label{c1}
    For ensembles of multipartite product states, one-way LOCC (in fact, LOCC without any communication) is optimal for antidistinguishability.
\end{coro}

\begin{lemma}\label{l3}
    One-way LOCC antidistinguishability of a set of multipartite product states does not depend on the starter of the protocol.
\end{lemma}
\begin{proof}
If the set is LOCC antidistinguishable, then by Theorem \ref{th2} some party $j$ can antidistinguish her local states. In a one-way protocol with any ordering of the parties, all parties other than $j$ perform trivial measurements and party $j$ performs the antidistinguishing measurement. If the set is not LOCC antidistinguishable, no protocol succeeds, whatever the ordering.
\end{proof}

This is a departure from the case of LOCC distinguishability of product states, where an asymmetry depending on the starter of the protocol exists \cite{hardy}.  

Now we extend this question to the one-way LOCC $2$-antidistinguishability task. Consider the following product states acting on $\mathbbm{C}^2\otimes\mathbbm{C}^2$.
\bea\label{x_1}
\ket{\delta_1}&=&\ket{0}\ket{0},
\ket{\delta_2}=\ket{+}\ket{1},\nonumber\\
\ket{\delta_3}&=&\ket{v_+}\ket{+},
\ket{\delta_4}=\ket{v_-}\ket{-}.
\eea

\begin{thm}\label{th3}
     One-way LOCC $2$-antidistinguishability of product states can depend on the starter: the states defined in \eqref{x_1} are one-way LOCC $2$-antidistinguishable if Alice starts the protocol, but not if Bob starts.
\end{thm}
\begin{proof}
   Without communication, Alice cannot eliminate two states in a single measurement outcome, because her four local states are pairwise linearly independent qubit vectors and it is not possible to construct a nonzero POVM element that is simultaneously orthogonal to two of them. Hence, the strategy must involve one elimination on Alice's side, followed by another elimination among the remaining states on Bob's side. Since the states $\{\ket{0}, \ket{+}, \ket{v_+}, \ket{v_-}\}$ form an antidistinguishable set, Alice can always eliminate one state corresponding to each measurement outcome. 
   
A suitable set of POVM elements is 
\[
\Big\{ \frac12\ket{1}\bra{1},
\frac12\ket{-}\bra{-},
\frac{\sqrt{3}-1}{2\sqrt{3}}\ket{v_+^\perp}\bra{v_+^\perp},
\frac{\sqrt{3}+1}{2\sqrt{3}}\ket{v_-^\perp}\bra{v_-^\perp}\Big\}.
\]

Each measurement outcome of Alice reduces the remaining set at Bob's side to an antidistinguishable subset, ensuring that two states are excluded in every run of the protocol. Note that, in each run, Bob is left with three states, two of which are orthogonal. Therefore, Bob performs a measurement in the corresponding orthogonal basis, achieving $2$-antidistinguishability of the set. However, the protocol does not provide \emph{strong} $2$-antidistinguishability, as not all pairs ($2$-tuples) are exhausted. The following pairs are eliminated in this case:
\[
\{\delta_1,\delta_3\},\{\delta_1,\delta_4\},\{\delta_2,\delta_3\},\{\delta_2,\delta_4\}.
\]

If Bob starts the protocol instead, Bob needs to eliminate one state from his share for the same reason as Alice: an outcome of Bob eliminating no state would require Alice to eliminate two states with a single element of her qubit measurement, which is impossible. As every nonzero element eliminating one of Bob's four qubit states is proportional to the projector onto the orthogonal vector, his measurement has the form
\[
\left\{g_1\ket{0}\bra{0},\; g_2\ket{1}\bra{1},\; g_3\ket{+}\bra{+},\; g_4\ket{-}\bra{-}\right\},
\]
with $g_i\geq 0$ and $\sum_i g_i=2$. If the second outcome clicks ($\ket{\delta_1}$ eliminated), Alice must antidistinguish between $\ket{+}$, $\ket{v_+}$, and $\ket{v_-}$, whose squared overlaps $\frac{2+\sqrt3}{4},\frac{2-\sqrt3}{4},\frac14$ sum to more than one, so they are not antidistinguishable by \eqref{condAS}. Hence, this run of the protocol fails to achieve $2$-antidistinguishability. Similarly, when the third outcome occurs ($\ket{\delta_4}$ eliminated), the resulting set $\{\ket{0},\ket{+},\ket{v_+}\}$ at Alice's side is again not antidistinguishable. If these two outcomes never occur, i.e., $g_2=g_3=0$, then $g_1\ket{0}\bra{0}+g_4\ket{-}\bra{-}\neq\I_2$, so the measurement of Bob is not a valid one. Therefore, the set is not one-way LOCC $2$-antidistinguishable when Bob initiates the protocol.

\end{proof}
By the remark in Sec. III, in a multi-round protocol Bob could open with a trivial measurement and let Alice run the above protocol, so the asymmetry of Theorem \ref{th3} is a property of one-way protocols.

A similar phenomenon, now with Alice's protocol achieving the strong notion, is exhibited by the following set of states on $\mathbbm{C}^2\otimes \mathbbm{C}^4$:
\bea\label{x_anti}
\ket{\phi_1}&=&\ket{0}\ket{0}_4\nonumber\\
\ket{\phi_2}&=&\ket{+}(\epsilon\ket{0}_4+\sqrt{1-\epsilon^2}\ket{1}_4)\nonumber\\
\ket{\phi_3}&=&\ket{v_+}\left(\epsilon\ket{0}_4+\frac{\epsilon-\epsilon^2}{\sqrt{1-\epsilon^2}}\ket{1}_4\right.\nonumber\\
&&\left.+\sqrt{\frac{(1-\epsilon)(1+2\epsilon)}{1+\epsilon}}\ket{2}_4\right)\nonumber\\
\ket{\phi_4}&=&\ket{v_-}\left(\epsilon\ket{0}_4+\frac{\epsilon-\epsilon^2}{\sqrt{1-\epsilon^2}}\ket{1}_4\right.\nonumber\\
&+&\epsilon\sqrt{\frac{(1-\epsilon)}{(1+\epsilon)(1+2\epsilon)}}\ket{2}_4\nonumber\\
&+&\left.\sqrt{\frac{(1-\epsilon)(1+3\epsilon)}{1+2\epsilon}}\ket{3}_4\right),
\eea
where $\ket{i}_d$ is the computational basis vector $\ket{i}$ in dimension $d$ and $0<\epsilon<1$. We denote Bob's parts by $\ket{\phi^B_i}$. A special feature of Bob's states is that all pairwise inner products are equal to $\epsilon$, i.e., their Gram matrix is $(1-\epsilon)\I_4+\epsilon J$, where $J$ is the all-ones matrix.
\begin{proposition}\label{p1}
    Consider the set of states \eqref{x_anti} with $\frac{1}{\sqrt5}<\epsilon\leqslant\frac12$.
    If Alice starts the protocol, these states are strongly one-way LOCC $2$-antidistinguishable.
    If Bob starts, the states are not one-way LOCC $2$-antidistinguishable and, a fortiori, not strongly one-way LOCC $2$-antidistinguishable.
\end{proposition}
\begin{proof}
    If Alice starts, this part holds for all $0<\epsilon\leqslant\frac12$. Alice performs the measurement described in Theorem \ref{th3}, every outcome of which eliminates exactly one state. 
For each outcome of Alice's measurement, Bob is required to antidistinguish among the three remaining states on his side, all of whose pairwise inner products are equal to $\epsilon$. 
According to Ref.~\cite{barett}, three states are antidistinguishable when $\epsilon \leqslant \frac{1}{2}$. 
Since these three states are pairwise non-orthogonal, every antidistinguishing measurement has three non-null elements. Hence, if Alice eliminates $\ket{\phi_k}$, Bob subsequently eliminates each of the other three states in some run. 
This process exhausts all possible pairwise eliminations, implying that the set of states is strongly $2$-antidistinguishable.

If Bob starts, we show that, for $0<\epsilon\leqslant\frac12$, a one-way protocol starting with Bob exists if and only if $5\epsilon^2\leqslant1$, and that in this case it achieves strong $2$-antidistinguishability. Alice measures last, and her local states $\ket{0},\ket{+},\ket{v_+},\ket{v_-}$ are pairwise linearly independent qubit vectors, so every nonzero element of Alice annihilates at most one of them. For an element $M$ of Bob, let $K(M)=\{i: M\ket{\phi^B_i}=0\}$. If $K(M)=\emptyset$, Alice would have to eliminate two states with a single element, which is impossible. If $K(M)=\{1\}$ ($K(M)=\{4\}$), Alice must antidistinguish $\{\ket{+},\ket{v_+},\ket{v_-}\}$ ($\{\ket{0},\ket{+},\ket{v_+}\}$), whose squared overlaps $\frac{2+\sqrt3}{4},\frac{2-\sqrt3}{4},\frac14$ ($\frac12,\frac14,\frac{2+\sqrt3}{4}$) sum to more than one; this is impossible by \eqref{condAS}. If $K(M)=\{2\}$, Alice must antidistinguish the trine $\{\ket{0},\ket{v_+},\ket{v_-}\}$ of Example \ref{e1}, and if $K(M)=\{3\}$, she must antidistinguish $\{\ket{0},\ket{+},\ket{v_-}\}$, whose squared overlaps $\frac12,\frac14,\frac{2-\sqrt3}{4}$ satisfy \eqref{condAS} with equality in the second condition; both are possible. Grouping Bob's elements according to whether they annihilate $\ket{\phi^B_2}$, otherwise $\ket{\phi^B_3}$, otherwise both $\ket{\phi^B_1}$ and $\ket{\phi^B_4}$, we conclude that a one-way protocol starting with Bob exists if and only if there are positive semidefinite operators with
\bea\label{bobpovm}
&&M_2+M_3+M_{14}=\I_4,\nonumber\\
&&M_2\ket{\phi^B_2}=0,\quad M_3\ket{\phi^B_3}=0,\nonumber\\
&&M_{14}\ket{\phi^B_1}=M_{14}\ket{\phi^B_4}=0 .
\eea

First, for $R\geq0$, a decomposition $R=M_2+M_3$ with $M_2,M_3\geq0$, $M_2\ket{\phi^B_2}=0$ and $M_3\ket{\phi^B_3}=0$ exists if and only if $\bra{\phi^B_2}R\ket{\phi^B_3}=0$. Necessity is clear, since $\bra{\phi^B_2}R\ket{\phi^B_3}=\bra{\phi^B_2}M_2\ket{\phi^B_3}+\bra{\phi^B_2}M_3\ket{\phi^B_3}=0$. Conversely, if $\bra{\phi^B_2}R\ket{\phi^B_3}=0$, the vectors $\ket{u_i}=R^{1/2}\ket{\phi^B_i}$ ($i=2,3$) are orthogonal. Let $\Pi$ be the projector onto the support of $R$ and $Q$ the projector onto $\ket{u_2}$ ($Q=0$ if $\ket{u_2}=0$). Then $M_3=R^{1/2}QR^{1/2}$ and $M_2=R^{1/2}(\Pi-Q)R^{1/2}$ are positive semidefinite and sum to $R$, while $M_3\ket{\phi^B_3}=R^{1/2}Q\ket{u_3}=0$ and $M_2\ket{\phi^B_2}=R^{1/2}(\Pi\ket{u_2}-\ket{u_2})=0$.

Applying this to $R=\I_4-M_{14}$, the conditions \eqref{bobpovm} can be met if and only if there is an operator $0\leq M_{14}\leq\I_4$ annihilating $\ket{\phi^B_1}$ and $\ket{\phi^B_4}$ with $\bra{\phi^B_2}M_{14}\ket{\phi^B_3}=\la\phi^B_2|\phi^B_3\ra=\epsilon$. Let $P$ be the projector onto ${\rm span}\{\ket{\phi^B_1},\ket{\phi^B_4}\}$ and $\ket{w_i}=(\I_4-P)\ket{\phi^B_i}$. Then $\bra{\phi^B_2}M_{14}\ket{\phi^B_3}=\bra{w_2}M_{14}\ket{w_3}$. As all vectors are real, replacing $M_{14}$ by its entrywise real part preserves all constraints, so we may take $M_{14}$ real symmetric. For $0\leq M\leq\I$, $\bra{w_2}M\ket{w_3}=\tr(MH)$ with $H=\frac12(\ket{w_2}\bra{w_3}+\ket{w_3}\bra{w_2})$, whose nonzero eigenvalues are $\frac12(\la w_2|w_3\ra\pm\|w_2\|\|w_3\|)$. Hence the attainable values form an interval containing $0$ whose maximum is $\frac12(\|w_2\|\|w_3\|+\la w_2|w_3\ra)$. Using the Gram matrix, $\bra{\phi^B_a}P\ket{\phi^B_b}=(\epsilon,\epsilon)\begin{pmatrix}1&\epsilon\\ \epsilon&1\end{pmatrix}^{-1}(\epsilon,\epsilon)^T=\frac{2\epsilon^2}{1+\epsilon}$ for $a,b\in\{2,3\}$, so that
\be
\la w_2|w_3\ra=\frac{\epsilon(1-\epsilon)}{1+\epsilon},\qquad \|w_2\|^2=\|w_3\|^2=\frac{(1-\epsilon)(1+2\epsilon)}{1+\epsilon},
\ee
and the maximum equals $\frac{(1-\epsilon)(1+3\epsilon)}{2(1+\epsilon)}$. This is at least $\epsilon$ if and only if $(1-\epsilon)(1+3\epsilon)\geq2\epsilon(1+\epsilon)$, i.e., $5\epsilon^2\leq1$. We have also confirmed this threshold numerically by semidefinite programming. Therefore, for $\frac{1}{\sqrt5}<\epsilon\leqslant\frac12$, no one-way protocol starting with Bob achieves $2$-antidistinguishability, which proves the proposition.

Finally, when $5\epsilon^2\leq1$, all three operators in \eqref{bobpovm} are non-null: $M_{14}\neq\mathbf{0}$ since otherwise $\la\phi^B_2|\phi^B_3\ra=0$; and if $M_2=\mathbf{0}$, then $M_3+M_{14}=\I_4$ and $M_3\ket{\phi^B_3}=0$ give $M_{14}\ket{\phi^B_3}=\ket{\phi^B_3}$, whence $\epsilon=\la\phi^B_1|\phi^B_3\ra=\bra{\phi^B_1}M_{14}\ket{\phi^B_3}=0$, a contradiction (similarly $M_3\neq\mathbf{0}$). After the outcome $M_2$ ($M_3$), Alice antidistinguishes a pairwise non-orthogonal triple, so all three of her elements are non-null and the pairs $\{\phi_2,\phi_1\},\{\phi_2,\phi_3\},\{\phi_2,\phi_4\}$ ($\{\phi_3,\phi_1\},\{\phi_3,\phi_2\},\{\phi_3,\phi_4\}$) are eliminated, while the outcome $M_{14}$ eliminates $\{\phi_1,\phi_4\}$. Hence all six pairs are exhausted.
\end{proof}

For comparison, Bob's local states alone are not $2$-antidistinguishable for $\epsilon>\frac13$ (Appendix \ref{B}). By suitably adjusting the parameter $\epsilon$, the states in \eqref{x_anti} can be made strongly one-way LOCC $2$-antidistinguishable irrespective of the commencing party.
\begin{proposition}\label{p2}
    For $0<\epsilon\leqslant\frac{1}{\sqrt5}$, the states \eqref{x_anti} are strongly one-way LOCC $2$-antidistinguishable irrespective of the starter. If, moreover, $\epsilon\leqslant\frac13$, Bob's local states are themselves strongly $2$-antidistinguishable, so that no communication is needed when Bob starts.
\end{proposition}
\begin{proof}
   The proof of Proposition~\ref{p1} establishes that the set of states is strongly one-way LOCC $2$-antidistinguishable when Alice initiates the protocol for $0<\epsilon \leqslant \tfrac{1}{2}$, and when Bob initiates it for $0<\epsilon\leqslant\frac{1}{\sqrt5}$. 
Furthermore, a semidefinite program confirms that the states on Bob's side are strongly $2$-antidistinguishable when $\epsilon\leqslant\frac13$ (Appendix \ref{B}). 
Therefore, in this range Bob can perform strong $2$-antidistinguishability of the set without any communication.
\end{proof}
\section{Nonlocality without entanglement}
Now we move to the results concerning the relation between local and global exclusion tasks, which will eventually lead us to our main goal of this work, i.e., \emph{nonlocality without entanglement}. The first theorem in this direction provides a minimal set of states which shows this peculiar quantum property.
\begin{thm}\label{th4}
    For every $d\geq2$ there exist three product states in $\mathbbm{C}^d\otimes\mathbbm{C}^d$ that are globally (strongly) antidistinguishable but not LOCC antidistinguishable.
\end{thm}
\begin{proof}
Let $\{\ket{0},\ket{\xi_1},\cdots,\ket{\xi_{d-1}}\}$ be an orthonormal basis of $\mathbbm{C}^d$ and let
\be
\ket{v}=\tfrac{1}{2}\ket{0} + \sum_{i=1}^{d-1} r_i \ket{\xi_i},\qquad \sum_{i=1}^{d-1}|r_i|^2=\tfrac34,
\ee
so that $\la 0|v\ra=\frac12$. Consider the following three states:
\be\label{th4states}
 \ket{0}_d\ket{0}_d, \qquad
 \ket{0}_d\ket{v}, \qquad
 \ket{v}\ket{v}.
\ee
Their squared overlaps are $x_1=\frac14$, $x_2=\frac1{16}$ and $x_3=\frac14$, so $x_1+x_2+x_3=\frac{9}{16}<1$ and $(x_1+x_2+x_3-1)^2=\frac{49}{256}\geq4x_1x_2x_3=\frac1{64}$. Hence the states satisfy the necessary and sufficient conditions \eqref{condAS} and are globally antidistinguishable; being pairwise non-orthogonal, they are strongly antidistinguishable.

The local lists are $\{\ket{0},\ket{0},\ket{v}\}$ for Alice and $\{\ket{0},\ket{v},\ket{v}\}$ for Bob. Neither is antidistinguishable: a POVM whose elements each annihilate $\ket{0}$ or $\ket{v}$ can be grouped as $E_0+E_v=\I$ with $E_0\ket{0}=0$ and $E_v\ket{v}=0$; then $E_v\ket{0}=\ket{0}$, and hence $\la v|0\ra=\bra{v}E_v\ket{0}=0$, a contradiction. By Theorem~\ref{th2}, the set of states is not LOCC antidistinguishable.
\end{proof}
This is our first result demonstrating \emph{nonlocality without entanglement} in the context of antidistinguishability of states. 
The minimal number of states required to exhibit this phenomenon is three: 
for two states, distinguishability and antidistinguishability coincide, and any two orthogonal pure states can be perfectly distinguished by LOCC \cite{vedral}. 
Here, we extend this phenomenon to the case of $2$-antidistinguishability.

Consider the following four quantum states acting on $\mathbbm{C}^2\otimes\mathbbm{C}^2$:\\
\bea\label{pbr}
&&\ket{+\theta}\ket{+\theta},
\ket{+\theta}\ket{-\theta},\nonumber\\
&&\ket{-\theta}\ket{+\theta},
\ket{-\theta}\ket{-\theta},
\eea
where $\ket{+\theta} = \cos\theta\ket{0}+\sin\theta\ket{1}$ and $\ket{-\theta} = \cos\theta\ket{0}-\sin\theta\ket{1}$.
\begin{thm}\label{th6}
  The states described in \eqref{pbr} with $0<\cos2\theta\leqslant \sqrt{2}-1$ are globally strongly $2$-antidistinguishable but not LOCC $2$-antidistinguishable. In fact, for every $\theta$ with $\cos2\theta\neq0$, they are not even LOCC antidistinguishable.
\end{thm}
\begin{proof}
   Ref.~\cite{webb} showed that these states are strongly $2$-antidistinguishable with a $6$-outcome entangled measurement if $\theta$ lies in the given range; at the boundary $\cos2\theta=\sqrt2-1$ an explicit measurement is given in Example \ref{e4}. We have also confirmed this range numerically by solving the semidefinite program corresponding to \eqref{pA2b}.

   The local lists are $\{\ket{+\theta},\ket{+\theta},\ket{-\theta},\ket{-\theta}\}$ for Alice and $\{\ket{+\theta},\ket{-\theta},\ket{+\theta},\ket{-\theta}\}$ for Bob. Neither is antidistinguishable when $\la+\theta|-\theta\ra=\cos2\theta\neq0$: grouping the elements of a putative antidistinguishing POVM as $E_++E_-=\I$ with $E_\pm\ket{\pm\theta}=0$ gives $E_-\ket{+\theta}=\ket{+\theta}$, and hence $\la-\theta|+\theta\ra=\bra{-\theta}E_-\ket{+\theta}=0$, a contradiction. By Lemma \ref{l5}, every finite-round LOCC protocol has an outcome that eliminates no state at all. Hence the states are not LOCC antidistinguishable and, a fortiori, not LOCC $2$-antidistinguishable. Note that for $\cos2\theta=0$ the states form an orthogonal product basis and are trivially LOCC distinguishable, which is why this value is excluded.
\end{proof}

Using the same set of states, we can show that the phenomenon of \emph{nonlocality without entanglement} exists under antidistinguishability too.
\begin{thm}\label{th6b}
    The states in \eqref{pbr} with $2\theta=\frac{\pi}{4}$ are globally strongly antidistinguishable but not LOCC antidistinguishable.
\end{thm}
\begin{proof}
   For $\theta=\frac{\pi}{8}$, $\la+\theta|-\theta\ra=\cos\frac{\pi}{4}=\la0|+\ra$, so a local (real orthogonal) transformation on each qubit maps the states \eqref{pbr} to $\{\ket{0}\ket{0},\ket{0}\ket{+},\ket{+}\ket{0},\ket{+}\ket{+}\}$. The work in Ref.~\cite{Pusey_2012} established that this set is antidistinguishable by a measurement in an entangled orthonormal basis with four non-null elements, so the set of states \eqref{pbr} is strongly antidistinguishable.
   
Since $\cos2\theta\neq0$, Theorem \ref{th6} (equivalently, Theorem \ref{th2}, as neither party's local list is antidistinguishable) shows that the set is not LOCC antidistinguishable.
\end{proof}
 Moving on, now we prove another version of the nonlocality argument. For that reason,
consider the following four $2\otimes 2$ states:
\bea\label{n2}
\ket{\zeta_1}&=&\ket{0}\ket{0},
\ket{\zeta_2}=\ket{1}\ket{1}\nonumber\\
\ket{\zeta_3}&=&\ket{+}\ket{\eta_1},
\ket{\zeta_4}=\ket{-}\ket{\eta_2},
\eea
where $\ket{\eta_1}=\cos\frac{\pi}{6}\ket{0}+\sin\frac{\pi}{6}\ket{1}$ and $\ket{\eta_2}=\cos\frac{\pi}{12}\ket{0}+\sin\frac{\pi}{12}\ket{1}$.
\begin{proposition}\label{th7}
    The states of \eqref{n2} are globally strongly $2$-antidistinguishable and strongly LOCC antidistinguishable (without communication), but, within the one-way paradigm, they are not strongly LOCC $2$-antidistinguishable, irrespective of the starter.
\end{proposition}
\begin{proof}
    Strong global $2$-antidistinguishability is verified by the semidefinite program described in Appendix \ref{C}. These states are strongly LOCC antidistinguishable as Alice has the strongly antidistinguishable set of states $\Big\{\ket{0},\ket{1},\ket{+},\ket{-}\Big\}$ (see Example \ref{e3}), each of whose four outcomes eliminates a different state $\ket{\zeta_i}$.
    
    Both local lists consist of four pairwise linearly independent qubit vectors, so every nonzero local element annihilates at most one of them. In a one-way protocol, the first party must therefore eliminate exactly one state in each outcome (an outcome eliminating nothing would require the second party to eliminate two states with a single element), and the second party must then antidistinguish the remaining three local states.

     Suppose Alice starts the protocol. If Alice eliminates $\ket{\zeta_1}$ (element $\propto\ket{1}\bra{1}$), Bob has to antidistinguish $\{\ket{1}, \ket{\eta_1}, \ket{\eta_2}\}$, whose squared overlaps $\frac14,\frac{2-\sqrt3}{4},\frac{2+\sqrt3}{4}$ sum to more than one; this is not possible by \eqref{condAS}. The same is true if Alice eliminates $\ket{\zeta_2}$, since $\{\ket{0}, \ket{\eta_1}, \ket{\eta_2}\}$ has squared overlaps $\frac34,\frac{2+\sqrt3}{4},\frac{2+\sqrt3}{4}$. Hence Alice's measurement consists of multiples of $\ket{-}\bra{-}$ (eliminating $\ket{\zeta_3}$) and $\ket{+}\bra{+}$ (eliminating $\ket{\zeta_4}$), i.e., it is the measurement in the $\{\ket{+},\ket{-}\}$ basis. After eliminating $\ket{\zeta_3}$, Bob must antidistinguish $\{\ket{0},\ket{1},\ket{\eta_2}\}$; a POVM of the form $a\ket{1}\bra{1}+b\ket{0}\bra{0}+c\ket{\eta_2^\perp}\bra{\eta_2^\perp}=\I_2$ requires $c=0$ because of the off-diagonal entries, so Bob measures in the computational basis and eliminates $\ket{\zeta_1}$ or $\ket{\zeta_2}$. The same holds after eliminating $\ket{\zeta_4}$. Henceforth, the eliminated pairs are $\Big\{\{\ket{\zeta_1},\ket{\zeta_3}\},\{\ket{\zeta_2},\ket{\zeta_3}\},\{\ket{\zeta_1},\ket{\zeta_4}\}, \{\ket{\zeta_2},\ket{\zeta_4}\}\Big\}$, and the pairs $\{\ket{\zeta_1},\ket{\zeta_2}\}$ and $\{\ket{\zeta_3},\ket{\zeta_4}\}$ are never eliminated. So the states are not strongly one-way LOCC $2$-antidistinguishable when Alice starts.

   If Bob starts, the set $\{\ket{0},\ket{1},\ket{\eta_1},\ket{\eta_2}\}$ is antidistinguishable but not strongly antidistinguishable. Bob's elements are multiples of $\ket{1}\bra{1}$, $\ket{0}\bra{0}$, $\ket{\eta_1^\perp}\bra{\eta_1^\perp}$ and $\ket{\eta_2^\perp}\bra{\eta_2^\perp}$, eliminating $\ket{\zeta_1}$, $\ket{\zeta_2}$, $\ket{\zeta_3}$ and $\ket{\zeta_4}$, respectively. If Bob eliminates $\ket{\zeta_1}$ or $\ket{\zeta_2}$, Alice must antidistinguish $\{\ket{1},\ket{+},\ket{-}\}$ or $\{\ket{0},\ket{+},\ket{-}\}$; writing her POVM as $a\ket{0}\bra{0}+b\ket{-}\bra{-}+c\ket{+}\bra{+}=\I_2$ (or with $\ket{1}\bra{1}$ in place of $\ket{0}\bra{0}$) forces $b=c=1$ and $a=0$, so Alice measures in the $\{\ket{+},\ket{-}\}$ basis and eliminates $\ket{\zeta_3}$ or $\ket{\zeta_4}$. If Bob eliminates $\ket{\zeta_3}$ or $\ket{\zeta_4}$, Alice must antidistinguish $\{\ket{0},\ket{1},\ket{-}\}$ or $\{\ket{0},\ket{1},\ket{+}\}$, which similarly forces the computational-basis measurement, eliminating $\ket{\zeta_1}$ or $\ket{\zeta_2}$. Thus the eliminated pairs are contained in $\Big\{\{\ket{\zeta_1},\ket{\zeta_3}\},\{\ket{\zeta_1},\ket{\zeta_4}\},\{\ket{\zeta_2},\ket{\zeta_3}\}, \{\ket{\zeta_2},\ket{\zeta_4}\}\Big\}$.\\
   Thus, irrespective of the starter, the set of states is not strongly one-way LOCC $2$-antidistinguishable.
\end{proof}
An even stronger notion of nonlocality emerges with genuinely nonlocal sets of states, where an $N$-partite set of states shows some global property which is unachievable even if $(N-1)$ parties come together. It is a well-explored field in the context of distinguishability of states \cite{Halder,gupta,ghosh_s}. We demonstrate an analogous notion in the case of antidistinguishability. For this proof, let us
consider the following three quantum states acting on $\mathbbm{C}^2\otimes \mathbbm{C}^2\otimes \mathbbm{C}^2$:\\
\bea\label{pr}
&&\ket{0}\ket{0}\ket{0}\nonumber\\
&&\ket{0}\ket{+}\ket{+}\nonumber\\
&&\ket{+}\ket{+}\ket{0}
\eea

\begin{thm}\label{th8}
     The states described in \eqref{pr} are globally (strongly) antidistinguishable but not LOCC antidistinguishable across any bipartition.
\end{thm}
\begin{proof}
     All pairwise inner products of the states equal $\frac12$, so $x_1=x_2=x_3=\frac14$; then $x_1+x_2+x_3=\frac34<1$ and $(x_1+x_2+x_3-1)^2=\frac1{16}=4x_1x_2x_3$. The states satisfy the conditions \eqref{condAS} and are therefore globally antidistinguishable, and strongly so since they are pairwise non-orthogonal.

    There exist three possible bipartitions of the system, and the states are product across each of them. For the cut $A|BC$, the local lists are $\{\ket{0},\ket{0},\ket{+}\}$ and $\{\ket{0}\ket{0},\ket{+}\ket{+},\ket{+}\ket{0}\}$; for $B|AC$, they are $\{\ket{0},\ket{+},\ket{+}\}$ and $\{\ket{0}\ket{0},\ket{0}\ket{+},\ket{+}\ket{0}\}$; for $C|AB$, they are $\{\ket{0},\ket{+},\ket{0}\}$ and $\{\ket{0}\ket{0},\ket{0}\ket{+},\ket{+}\ket{+}\}$. Each single-qubit list contains only the two non-orthogonal states $\ket{0}$ and $\ket{+}$ and is not antidistinguishable, by the argument in the proof of Theorem \ref{th4}. Each two-qubit list is pairwise non-orthogonal with squared overlaps $\frac14,\frac12,\frac12$, whose sum exceeds one, so it is not antidistinguishable by \eqref{condAS}. According to Theorem \ref{th2}, LOCC antidistinguishability across a cut requires that at least one side possesses an antidistinguishable list of local states. Hence the states are not LOCC antidistinguishable across any bipartition and, since tripartite LOCC is a subset of LOCC across any cut, not tripartite LOCC antidistinguishable either.
\end{proof}
Needless to say, one can extend this notion to more than three parties using the same strategy.
\section{Conclusion}

In this work, we investigated the problem of exclusion of quantum states, with particular emphasis on the limitations imposed by LOCC protocols. 
We first analyzed the structure of LOCC antidistinguishability and identified a sufficient, but not necessary, condition under which three bipartite states can be locally antidistinguished. We then proved that a set of multipartite product states is LOCC antidistinguishable if and only if one of the parties can antidistinguish her local states alone. Consequently, unlike state distinguishability, LOCC antidistinguishability of product states requires no communication and does not exhibit any asymmetry with respect to the initiating party. This symmetry breaks down for one-way LOCC $2$-antidistinguishability, for which we presented explicit product sets whose exclusion depends on which party starts the protocol.
Our main results establish the existence of \emph{nonlocality without entanglement} in the context of state exclusion. We proved that three bipartite product states can be globally antidistinguishable while failing to be LOCC antidistinguishable, thereby identifying the minimal number of states required for this phenomenon. We further extended this property to the case of $2$-antidistinguishability, where global protocols outperform LOCC strategies with any number of rounds. In addition, we presented examples exhibiting strong and genuine forms of nonlocality, including tripartite product states that remain globally antidistinguishable but are not LOCC antidistinguishable across any bipartition.

These results highlight fundamental differences between distinguishability and antidistinguishability as operational tasks and establish exclusion-based nonlocality as a distinct manifestation of quantum nonlocality. Our work also raises several interesting open questions. In particular, it remains to determine whether the states of Proposition~\ref{th7} can be strongly $2$-antidistinguished by multi-round LOCC protocols. Note that the dependence on the commencing party in Theorem~\ref{th3} and Propositions~\ref{p1} and \ref{p2} is intrinsically a one-way notion, since in a multi-round protocol any party can open with a trivial measurement. More generally, a complete characterization of optimal LOCC protocols for higher-order $x$-antidistinguishability, in particular a characterization of LOCC $2$-antidistinguishability of product states analogous to Theorem~\ref{th2}, remains an important open problem. Other promising directions include understanding the role of shared entanglement and separable measurements in narrowing the gap between global and LOCC strategies, as well as exploring applications of exclusion-based nonlocality in quantum communication and cryptographic protocols. We hope that our findings motivate further investigation of quantum state exclusion as a fundamental operational resource.

\appendix

\section{SDP for Bob's local states in Eq.~\eqref{x_anti}}\label{B}
In this semidefinite program, we compute the $2$-antidistinguishability of the set of states $\{\ket{\phi_i^B}\}_i$, where the superscript $B$ denotes Bob's part of the states $\{\ket{\phi_i}\}_i$ in \eqref{x_anti}. Following \eqref{pA2b}, we get the following optimization problem:
\bea\label{b1}
\mathcal{S'}&=&\max_{\{N_{ij}\}}\left[1-\frac14\tr\big\{\left(\ket{\phi_1^B}\bra{\phi_1^B}N_{12} + \ket{\phi_2^B}\bra{\phi_2^B}N_{12}\right)\right.\nonumber\\
&& +\left. \left(\ket{\phi_1^B}\bra{\phi_1^B}N_{13} + \ket{\phi_3^B}\bra{\phi_3^B}N_{13}\right)\right.\nonumber\\
&& +\left. \left(\ket{\phi_1^B}\bra{\phi_1^B}N_{14} + \ket{\phi_4^B}\bra{\phi_4^B}N_{14}\right)\right.\nonumber\\
&& +\left. \left(\ket{\phi_2^B}\bra{\phi_2^B}N_{23}+ \ket{\phi_3^B}\bra{\phi_3^B}N_{23}\right)\right.\nonumber\\
&& +\left. \left(\ket{\phi_2^B}\bra{\phi_2^B}N_{24} + \ket{\phi_4^B}\bra{\phi_4^B}N_{24}\right)\right.\nonumber\\
&& +\left. \left(\ket{\phi_3^B}\bra{\phi_3^B}N_{34} + \ket{\phi_4^B}\bra{\phi_4^B}N_{34}\right)
\big\}\right]\nonumber\\
&\text{s.t.}& N_{ij}\geqslant 0\ \ \forall i,j,\nonumber\\
&& N_{12}+N_{13}+N_{14}+N_{23}+N_{24}+N_{34} =\I_4.\nonumber\\
\eea
Solving the above semidefinite program, we deduce that the states $\{\ket{\phi_i^B}\}_i$ are not $2$-antidistinguishable, i.e., $\mathcal{S'}< 1$, when $\frac13<\epsilon\leqslant\frac12$. For $\epsilon=\frac12$, the optimum POVM operators, in the basis $\{\ket{0}_4,\ket{1}_4,\ket{2}_4,\ket{3}_4\}$, are the following:
\[
N_{12} =
\begin{pmatrix}
0.0053 & 0.0031 & 0.0467 & 0.0362 \\
0.0031 & 0.0018 & 0.0270 & 0.0209 \\
0.0467 & 0.0270 & 0.4123 & 0.3193 \\
0.0362 & 0.0209 & 0.3193 & 0.2474
\end{pmatrix}
\]
\[
N_{13} =
\begin{pmatrix}
0.0053 & 0.0451 & -0.0127 & 0.0362 \\
0.0451 & 0.3836 & -0.108 & 0.308 \\
-0.0127 & -0.108 & 0.0304 & -0.0867 \\
0.0362 & 0.308 & -0.0867 & 0.2474
\end{pmatrix}
\]
\[
N_{14} =
\begin{pmatrix}
0.0053 & 0.0451 & 0.0319 & -0.0213 \\
0.0451 & 0.3836 & 0.2712 & -0.1816 \\
0.0319 & 0.2712 & 0.1918 & -0.1284 \\
-0.0213 & -0.1816 & -0.1284 & 0.086
\end{pmatrix}
\]
\[
N_{23} =
\begin{pmatrix}
0.328 & -0.1413 & -0.0999 & 0.2849 \\
-0.1413 & 0.0608 & 0.043 & -0.1227 \\
-0.0999 & 0.043 & 0.0304 & -0.0867 \\
0.2849 & -0.1227 & -0.0867 & 0.2474
\end{pmatrix}
\]
\[
N_{24} =
\begin{pmatrix}
0.328 & -0.1413 & 0.2508 & -0.1679 \\
-0.1413 & 0.0608 & -0.108 & 0.0723 \\
0.2508 & -0.108 & 0.1918 & -0.1284 \\
-0.1679 & 0.0723 & -0.1284 & 0.086
\end{pmatrix}
\]
\[
N_{34} =
\begin{pmatrix}
0.328 & 0.1894 & -0.2168 & -0.1679 \\
0.1894 & 0.1093 & -0.1252 & -0.097 \\
-0.2168 & -0.1252 & 0.1433 & 0.111 \\
-0.1679 & -0.097 & 0.111 & 0.086
\end{pmatrix}
\]

For Proposition \ref{p2}, we solve the same semidefinite program for $\epsilon\leqslant\frac13$ and obtain $\mathcal{S'}=1$ with an optimal measurement whose six elements are all non-null. For $\epsilon=\frac13$, the optimum measurement has the following elements:
\[
F_{12} =
\begin{pmatrix}
0 & 0 &  0 &  0 \\
0 & 0 &  0 &  0 \\
0 & 0 &  0.4 &  0.3265 \\
 0 & 0 &  0.3265 &  0.2667
\end{pmatrix}
\]

\[
F_{13} =
\begin{pmatrix}
0 &  0 & 0 &  0 \\
0 &  0.375 & -0.0968 &  0.3162 \\
0 & -0.0968 &  0.0250 & -0.0816 \\
0 &  0.3162 & -0.0816 &  0.2667
\end{pmatrix}
\]

\[
F_{14} =
\begin{pmatrix}
0 &  0 &  0 & 0 \\
0 &  0.3750 &  0.2904 & -0.1581 \\
0 &  0.2904 &  0.2250 & -0.1225 \\
0 & -0.1581 & -0.1225 &  0.0667
\end{pmatrix}
\]

\[
F_{23} =
\begin{pmatrix}
 0.3333 & -0.1178 & -0.0913 &  0.2981 \\
-0.1178 &  0.0417 &  0.0323 & -0.1054 \\
-0.0913 &  0.0323 &  0.025 & -0.0816 \\
 0.2981 & -0.1054 & -0.0816 &  0.2667
\end{pmatrix}
\]

\[
F_{24} =
\begin{pmatrix}
 0.3333 & -0.1178 &  0.2738 & -0.1491 \\
-0.1178 &  0.0417 & -0.0968 &  0.0527 \\
 0.2738 & -0.0968 &  0.2250 & -0.1225 \\
-0.1491 &  0.0527 & -0.1225 &  0.0667
\end{pmatrix}
\]

\[
F_{34} =
\begin{pmatrix}
 0.3333 &  0.2357 & -0.1826 & -0.1491 \\
 0.2357 &  0.1667 & -0.1291 & -0.1054 \\
-0.1826 & -0.1291 &  0.1000 &  0.0816 \\
-0.1491 & -0.1054 &  0.0816 &  0.0667
\end{pmatrix}
\]
 The element $F_{ij}$ eliminates $\ket{\phi_i^B}$ and $\ket{\phi_j^B}$ with $i,j\in\{1,\cdots,4\}$ and $i\neq j$.
 
\section{SDP of Proposition \ref{th7}}\label{C}
We have to exclude two states at each outcome out of four states. The problem statement is the same as \eqref{b1} with $\ket{\phi^B_i}$ replaced by $\ket{\zeta_i}$. Using semi-definite programming, we find $\mathcal{S}'=1$ and the following elements for the optimum measurement of that task, written in the basis $\{\ket{00},\ket{01},\ket{10},\ket{11}\}$:
\[
P_{12} =
\begin{pmatrix}
0 & 0 & 0 & 0 \\
0 & 0.0290 & -0.0032 & 0 \\
0 & -0.0032 & 0.0055 & 0 \\
0 & 0 & 0 & 0
\end{pmatrix}
\]
\[
P_{13} =
\begin{pmatrix}
0 & 0 & 0 & 0 \\
0 & 0.5873 & -0.0845 & -0.4408 \\
0 & -0.0845 & 0.0147 & 0.0592 \\
0 & -0.4408 & 0.0592 & 0.3384
\end{pmatrix}
\]
\[
P_{14} =
\begin{pmatrix}
0 & 0 & 0 & 0 \\
0 & 0.3232 & -0.0329 & 0.4460 \\
0 & -0.0329 & 0.0057 & -0.0540 \\
0 & 0.4460 & -0.0540 & 0.6476
\end{pmatrix}
\]
\[
P_{23} =
\begin{pmatrix}
0.5253 & -0.0598 & -0.4907 & 0 \\
-0.0598 & 0.0197 & 0.0485 & 0 \\
-0.4907 & 0.0485 & 0.4628 & 0 \\
0 & 0 & 0 & 0
\end{pmatrix}
\]
\[
P_{24} =
\begin{pmatrix}
0.4722 & 0.0657 & 0.4898 & 0 \\
0.0657 & 0.0250 & 0.0724 & 0 \\
0.4898 & 0.0724 & 0.5092 & 0 \\
0 & 0 & 0 & 0
\end{pmatrix}
\]
\[
P_{34}=
\begin{pmatrix}
0.0025 & -0.0059 &  0.0009 &  0 \\
-0.0059 &  0.0158 & -0.0003 & -0.0051 \\
0.0009 & -0.0003 &  0.0022 & -0.0051 \\
0 & -0.0051 & -0.0051 &  0.0140
\end{pmatrix}
\]
$P_{ij}$ is the element which excludes $\ket{\zeta_i}$ and $\ket{\zeta_j}$, where $i\neq j$ and $i,j\in\{1,\cdots,4\}$; up to numerical precision, $\bra{\zeta_i}P_{ij}\ket{\zeta_i}=\bra{\zeta_j}P_{ij}\ket{\zeta_j}=0$. Since all six elements are non-null, the set \eqref{n2} is strongly $2$-antidistinguishable.

\bibliography{ref} 
\end{document}